\documentclass[oribibl]{llncs}
\usepackage{llncsdoc}

\usepackage{graphicx}
\usepackage{verbatim}
\usepackage{amssymb}
\usepackage{amsmath}

\newtheorem{algo}[theorem]{Algorithm}

\newcommand{\mA}{{\cal A}}

\newcommand{\mC}{{\cal C}}

\newcommand{\mK}{{\cal K}}
\newcommand{\mL}{{\cal L}}

\newcommand{\mP}{{\cal P}}

\newcommand{\mR}{{\cal R}}
\newcommand{\mS}{{\cal S}}
\newcommand{\mT}{{\cal T}}

\begin{document}

\title{\Large On Finding Lekkerkerker-Boland Subgraphs}

\institute{Colorado State University, Fort Collins CO 80521, USA}

\author{Nathan Lindzey~\thanks{lindzey@cs.colostate.edu,
Computer Science Department,
Colorado State University,
Fort Collins, CO, 80523-1873
U.S.A.},
Ross M. McConnell~\thanks{rmm@cs.colostate.edu,
Computer Science Department,
Colorado State University,
Fort Collins, CO, 80523-1873
U.S.A.}
}

\date{}

\maketitle

\begin{abstract}
Lekkerkerker and Boland characterized the minimal forbidden induced subgraphs 
for the class of interval graphs.  We give a linear-time algorithm to
find one in any graph that is not an interval graph.  Tucker characterized the
minimal forbidden submatrices of matrices that do not have the consecutive-ones
property.  We give a linear-time algorithm to find one in any matrix that
does not have the consecutive-ones property.
\end{abstract}

\section{Introduction}

A graph is an {\em interval graph} if it is the intersection graph of a set of intervals
on a line.  Such a set of intervals is known as an {\em interval model} of the graph.
They are an important subclass of perfect graphs~\cite{Gol80}, they have been written
extensively about and they model constraints in various combinatorial optimization
and decision problems~\cite{Roberts78, SpinGR}.   They have a rich structure and history,
and interesting relationships to other graph classes.  For a survey, see~\cite{BLS99}.

If $M$ is a 0-1 (binary) matrix, we let $size(M)$ denote the number
of rows, columns and 1's.
Such a matrix has the {\em consecutive-ones property} if there exists a reordering
of its columns such that, in every row, the 1's are consecutive.
A {\em clique matrix} of a graph is a matrix that has a row for each
vertex, a column for each clique, and a 1 in row $i$, column $j$ if vertex $i$ is contained
in clique $j$.  
A graph is an interval graph if and only if its clique matrices have the consecutive-ones
property, see, for example~\cite{Gol80}.

In 1962, Lekkerkerker and Boland described the minimal induced forbidden subgraphs for
the class of interval graphs~\cite{LB62}, known as the {\em LB graphs} (Figure~\ref{fig:LBGraphs}).
Ten years later,
Tucker described the minimum forbidden submatrices for consecutive-ones matrices~\cite{Tuck72}.
These are depicted in Figure~\ref{fig:TuckGraphs}.  Not surprisingly, there is a relationship
between the intersection graphs of rows of Tucker matrices and the LB graphs, depicted in Figure~\ref{fig:LBGraphs}.

\begin{figure}
\centerline{\includegraphics[]{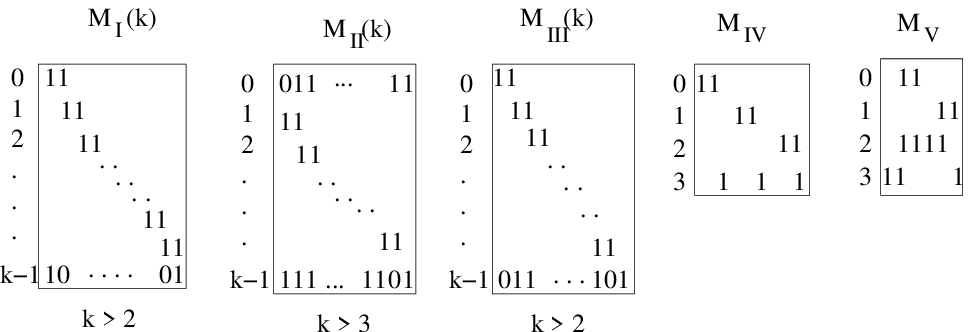}}
\caption{The Minimal Forbidden Submatrices for Consecutive-Ones Matrices.  For $M_I$, $k \geq 3$,
and for $M_{II}$ and $M_{II}$, $k \geq 4$.   $M_{IV}$ and $M_I$ have fixed size.}\label{fig:TuckGraphs}
\end{figure}

In this paper, we give a linear time bound for finding one of the LB
subgraphs when a graph is not an interval graph.   As part of our algorithm, we
also give a linear-time $(O(size(M))$ bound for finding one of Tucker's 
submatrices in a matrix $M$ that does not have the consecutive-ones property. 
This latter problem was solved previously in $O(n*size(M))$ time in~\cite{CCHS10}, where
$n$ is the number of rows of the matrix.

A graph is {\em chordal} if it has no {\em chordless cycle} ($G_{III}$ on four or more vertices).
Figure~\ref{fig:LBExtension} seems to imply that the rows of a given Tucker submatrix 
of the clique matrix can be extended to an interval model of an LB subgraph by including 
at most three additional rows of the matrix, giving the rows of vertices that induce
an LB subgraph.  Figure~\ref{fig:LBExtension} gives a 
counterexample.  Fortunately, it is true in the case of chordal graphs, but the example
illustrates the need for a proof; it does not follow from seemingly obvious considerations,
such as that no clique is a subset of any other.  This gives an LB subgraph if $G$ is
chordal, and when it is not chordal, the algorithm of~\cite{TarjYan85} already gives a $G_{III}$.

An interval graph is {\em proper} if there exists an interval representation where
no interval is a subset of another.  It is a {\em unit} interval graph if
there exists an interval representation where all intervals have the same length.
These graph classes are the same, and Wegner showed that a graph is a proper
interval graph if and only if it does not have a chordless cycle,
the special case of $G_{IV}$ or $G_V$
for $n = 6$ or the {\em claw} ($K_{1,3}$) as an induced subgraph~\cite{Wegner67}.
Hell and Huang give an algorithm that produces one of them in linear time~\cite{HellUIG} .
The problem of finding a forbidden subgraph reduces easily to finding an LB
subgraph.  Each of the LB graphs is either one of Wegner's forbidden subgraphs or contains
an obvious claw, and finding a claw in linear time, given an interval model, is elementary.
By itself, this approach has no obvious advantages over Hell and Huang's elegant
algorithm, but such reductions are useful when studying or programming a collection of related algorithms.

A {\em certifying algorithm} is an algorithm that provides, with each output, a simple-to-check
proof that it has answered correctly~\cite{KMMSJournal,MMNSCert}.  An interval model gives a 
certificate that a graph is an interval graph, and an LB subgraph gives one 
if the graph is not an interval graph.  However, a certifying algorithm was given previously
in~\cite{KMMSJournal}.  The ability to give a consecutive-ones ordering or a Tucker 
submatrix in linear time gives a linear-time certifying algorithm for consecutive-ones
matrices, but one was given previously in~\cite{McCSODA04}.  
The previous certificates are easier to check, which is a desirable
property for certifying algorithms.  However, they are neither minimal nor
uniquely characterized.  Aside from the theoretical interest in LB subgraphs, it is easy to obtain 
a minimal certificate of the form given in~\cite{KMMSJournal} from an LB subgraph found by
the algorithm we describe below.  An open question 
is whether a minimal, unique, especially simple, or otherwise interesting special case of the 
certificate from~\cite{McCSODA04} can be obtained by applying the algorithm of that paper to a Tucker 
submatrix obtained by the algorithm of the present paper.
Tucker submatrices may be useful in heuristics for finding large submatrices that have the consecutive-ones
property, small Tucker matrices, or identifying errors in biological data~\cite{StoyeWittler09, CCHS10}.  Our techniques
provide new tools for such heuristics.

\begin{figure}
\centerline{\includegraphics[]{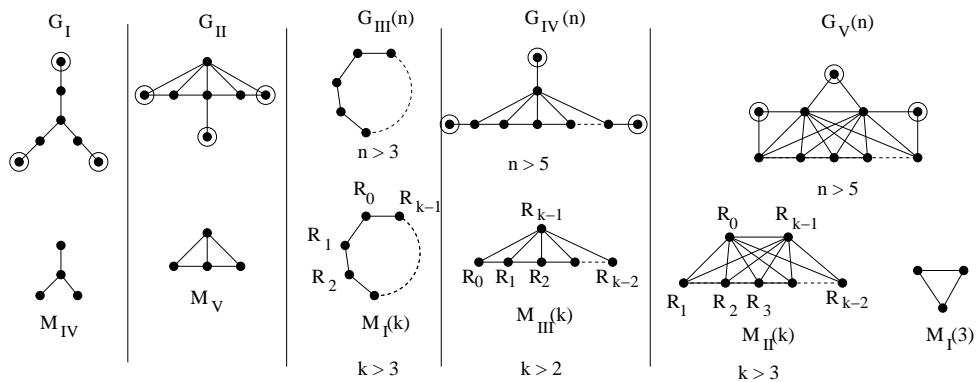}}
\caption{$G_I$ through $G_{V}$ are the minimal non-interval graphs discovered by
Lekkerkerker and Boland.  
Below them are the intersection graphs of the corresponding Tucker matrices.}\label{fig:LBGraphs}
\end{figure}

\section{Preliminaries}

Given a graph $G$, let $V$ denote the number of vertices and $E$ denote the number of edges.
If $\emptyset \subset X \subseteq V$, let $G[X]$ denote the subgraph induced by $X$.
Standard sparse representations of 0-1 matrices take $O(size(M))$ space
to represent $M$.  We treat the rows and columns as {\em sets}, where
each row $R$ is the set of columns where the row has a 1 and each column $C$
is the set of rows where the column has a 1.  
Suppose $\mR$ is the set of rows of a consecutive-ones ordered matrix
and $(C_1, C_2, \ldots, C_m)$ is the ordering of the columns.
In linear time, we can find, for each row, the leftmost and rightmost column
in the row.  Let us call these the {\em left endpoint} and {\em right endpoint} of the row.

That interval graphs are a subclass of the class of chordal graphs follows
from inclusion of $G_{III}$ among the LB subgraphs.  Rose, Tarjan and Lueker give an
$O(V+E)$ algorithm that finds whether a graph is a chordal graph, and, if so, produces its maximal
cliques~\cite{RTL:triangulated}.  Otherwise, the algorithm of~\cite{TarjYan85} produces
a chordless cycle ($G_{III}$) in linear time.

When a graph is chordal, the problem of deciding
whether it is an interval graph reduces to the problem of deciding whether its clique
matrix has the consecutive-ones property.   
Booth and Lueker further reduced this problem to
that of finding a maximal prefix $\mR' = \{R_1, R_2, \ldots, R_r\}$ of the rows of
a binary matrix $M$ that has the consecutive-ones property, in $O(size(M))$ time.

Assigning a left-to-right order to children of each internal node
of a rooted tree results in a unique left-to-right order of the leaves.
Booth and Lueker's algorithm produces a
{\em PQ tree}, for $\mR'$.  The PQ tree represents all possible consecutive-ones
orderings of $\mR'$.  There is one leaf $\{c\}$ for each column $c$.
The internal nodes of the PQ tree consists of {\em P nodes} and {\em Q nodes}.  
The consecutive-ones ordering of columns are given by the leaf
orders obtainable by assigning an arbitrary left-to-right order to 
children of each P node, and for each Q node, assigning the given 
left-to-right order or its reverse.  

Though the PQ tree can be represented using $O(1)$ space per node, 
conceptually, we will consider each node of the PQ tree
to be set given by the disjoint union of
its children; equivalently, it is the union of its leaf descendants.

\begin{definition}\label{lem:Venn}
Let $\mS$ be a collection of subsets of a set $U$.  Two elements of $U$ are in the
same {\em Venn class} if they are elements of the same set of members of $\mS$.
The {\em unconstrained Venn class} consists of those elements of $U$ that are not in
any member of $\mS$; all others are {\em constrained}.
Two sets $R_1, R_2$ {\em overlap} if their intersection is nonempty, but
neither is a subset of the other.  The {\em overlap graph} of $\mS$ is
the undirected graph whose vertices are the members of $\mS$, and $R_1, R_2 \in \mS$ are
adjacent if and only if $R_1$ and $R_2$ overlap.
\end{definition}

\begin{lemma}\label{lem:Meidanis}~\cite{MeidanisPT98}
A set of columns is a Q node of a consecutive-ones matrix
$M$ if and only if it is the union of rows of a connected
component of the overlap graph of rows of $M$.
The Venn classes of rows in this component are its children.
\end{lemma}

\section{Breadth-first search on the overlap graph of a collection of sets, given a consecutive-ones
ordering}\label{sect:BFS}

In linear time, we may label each row of a consecutive-ones ordered matrix with its
left and right endpoints.
We may then label each column $c_i$ of a consecutive-ones ordered matrix with
the set of rows that have their left endpoints in $c_i$.
In linear time, we can then radix sort the list of sets that have their left endpoint at $c_i$
in descending order of index of right endpoint, yielding a list $\mR_i$.  This is accomplished
with a single radix sort that has the index of the left endpoint as its primary sort key and index
of the right endpoint as the secondary sort key.  By symmetry, we can construct a list $\mL_i$
of rows that have their right endpoint in each column $c_i$, sorted in ascending order
of index of left endpoint.

This allows us to perform a breadth-first search on the overlap graph of the rows
in time linear in the size of the matrix, as follows.  The lists
$\mL_i$ and $\mR_i$ are represented with doubly-linked lists.
We maintain the invariant that elements that have been placed in
the BFS queue have been removed from these lists.  When a consecutive-ones ordered
set $R$ comes to the front of the queue, we traverse its list
$(c_j, c_{j+1}, \ldots, c_k)$ of
columns.  For each $c_h$ in the list,
we remove elements from $\mR_h$ and place them in the queue, until we reach an element in $\mR_i$
whose right endpoint is no farther to the right than $c_k$.  All of the removed elements
overlap $R$.  Since $\mR_h$ is sorted in descending order of right endpoint, all
these elements are a prefix of $\mR_h$, and any remaining elements in the list do
not overlap $R$.  When we remove an element from $\mR_h$, we remove it from any
list $\mL_{h'}$ that it is a member of, to maintain the invariant.  This takes
$O(1)$ time for each element moved to the BFS queue, plus $O(1)$ time
for each column of $R$.
The lists $L_h$ are handled symmetrically.
Summing over all rows $R$, the time is $O(size(M))$.

\section{Finding Tucker Submatrices}\label{sect:Tucker}

\subsection{Tucker matrices with at most four rows}

\begin{lemma}\label{lem:lastRowTucker}
If a set $\mR'$ of rows has the consecutive-ones property and $Z$ is a row such that
$\mR = \mR' \cup \{Z\}$ does not,
then $Z$ is one of the rows of every Tucker submatrix in $\mR$.
\end{lemma}

\begin{algo}\label{alg:fiveRows} 
(See Lemma~\ref{lem:fiveRowsCorrect}.)
\begin{tabbing}
init\=ialR\=ows\=($M', k$) \\
\>$M := M'$ \\
\>$i := 1$; \\
\>While $i \leq k+1$ and M has at least $i-1$ rows \\
\>\>   Using Booth and Lueker, find the minimal prefix $(R_1, R_2, \ldots, R_r, Z)$ of rows  \\
\>\>\>       of $M$ that does not have the consecutive-ones property. \\
\>\>   Let $M$ be the matrix whose row sequence is $(Z, R_1, R_2, \ldots, R_r)$. \\
\>return $M$ \\
\end{tabbing}
\end{algo}

\begin{lemma}\label{lem:fiveRowsCorrect}
Suppose Algorithm~\ref{alg:fiveRows} is run with parameter $k$ and a matrix $M'$ that does
not have the consecutive-ones property.  If the returned matrix $M$ has at most
$k$ rows, then these are the rows of every Tucker matrix of $M$.
Otherwise, $M$ fails to have the consecutive-ones property and every Tucker submatrix 
in $M$ has at least $k+1$ rows.
\end{lemma}
\begin{proof}
By induction on $i$, $M$ does not have the consecutive-ones
property at the end of iteration $i$.
Also, by induction on $i$, using Lemma~\ref{lem:lastRowTucker},
at the end of iteration $i$, for every Tucker submatrix $M_T$ of $M$, 
the rows of $M_T$ include the first $i$ rows of $M$.
If $M_T$ has only $i$ rows, then the first $i$ rows of $M$
do not have the consecutive-ones property, so at the end of iteration
$i+1$, $M$ will have $i$ rows.
\end{proof}

We run Algorithm~\ref{alg:fiveRows} for $k = 4$.  If it returns a matrix with $j$ rows, where
$j \leq 4$, it is easy to get a linear time bound to get the columns.
(One way is to generate all $j! \leq 24$ orderings of rows and for each, to check
for the columns of each Tucker matrix of size $j$.)
Otherwise, Algorithm~\ref{alg:fiveRows} returns a matrix $M$ of more than 4 rows.
By Lemma~\ref{lem:fiveRowsCorrect}, $M$
fails to have the consecutive-ones property and every Tucker
submatrix of $M$ has at least five rows.  This excludes any
instances of $M_{IV}$, $M_{V}$ or the anomalous case
of $M_{I}$ on three rows that does not correspond to a chordless cycle.

\subsection{Matrices in which all Tucker submatrices have more than four rows}

\begin{lemma}\label{lem:cycle}
The overlap graphs of $M_I$, $M_{II}$, and $M_{III}$ are simple cycles.
\end{lemma}

\begin{definition}\label{def:101}
Suppose $\mR'$ is a set of rows with the consecutive-one property, 
$Q$ is a Q node of its PQ tree, $(X_1, X_2, \ldots, X_k)$ is the ordering of $Q$'s 
children and $Z$ is a row not in $\mR'$. 
Let $X_h, X_i, X_j$ be three children
of $Q$ such that $h < i < j$.
They are a {\em 1-0-1 configuration for $Z$}
if $X_h$ and $X_j$ each contain a 1 of row $Z$
and $X_i$ contains a 0 of row $Z$.  They are a {\em 0-1-0 configuration}
for $Z$ if $X_h$ and $X_j$ each contain a 0 of row $Z$ and $X_i$ contains a 1 of row $Z$.
\end{definition}

\begin{lemma}\label{lem:violation}
If $\mR'$ is a set of rows that has the consecutive-ones property and $Z$ is a
row not in $\mR'$, then $\mR' \cup \{Z\}$ does not have the consecutive-ones
property if the PQ tree of $\mR'$ has a Q node $Q$ such that either:

\begin{enumerate}
\item \label{firstCase} $Q$ has a 1-0-1 configuration for $Z$;
\item\label{secondCase}  $Q$ has a 0-1-0 configuration for $Z$ and $Z$
is not a subset of $Q$.
\end{enumerate}
\end{lemma}

This test is implicit in Booth and Lueker's algorithm, where it
is a sufficient condition, but not a necessary one.   The following is
a consequence of Lemma~\ref{lem:Meidanis}.

\begin{lemma}\label{lem:violation2}~\cite{McCSODA04}
The conditions of Lemma~\ref{lem:violation} are necessary and sufficient
if the overlap graph of of $\mR'$ is connected.
\end{lemma}

\begin{lemma}\label{lem:corral}
If a matrix fails to have the consecutive-ones property and
has no Tucker submatrix with fewer than five rows, then when Algorithm~\ref{alg:fiveRows}
is run on it with $k=4$, at the end of one of the five iterations of its loop, the PQ
tree of $\mR' = \{R_1, R_2, \ldots, R_r\}$ will have a Q node $Q$ with the following
properties:

\begin{itemize}
\item $Q$ has a 1-0-1 configuration $(X_h, X_i, X_j)$ for $Z$;

\item There exist $A,B \in \mR'$ that are members of the
component of the overlap graph on $\mR'$ whose union is $Q$, and such that $A$
contains $X_h$ and is disjoint from $X_i$ and $X_j$, and $B$ contains $X_j$
and is disjoint from $X_h$ and $X_i$.
\end{itemize}
\end{lemma}
\begin{proof}
If $\mT$ is the rows of $M$ that contain a Tucker submatrix $M_T$,
at the end of an iteration of the loop of Algorithm~\ref{alg:fiveRows}, $Z \in \mT$
by by Lemma~\ref{lem:lastRowTucker}.
By Lemma~\ref{lem:cycle}, the overlap graph of $\mT' = \mT \setminus \{Z\}$ is
connected, so $\mT'$ is a subset of a component of the overlap graph
of $\mR'$, which gives rise to a Q node $Q$
of the PQ tree of $\mR'$, by Lemma~\ref{lem:Meidanis}.  Since the children of $Q$ are
the Venn classes of the component, no two Venn classes of $\mT'$, hence
no two columns of $M_T$, can lie in the same child of the Q node.

For each choice of a last row of a Tucker matrix on at least five rows,
Figure~\ref{fig:TuckRotations} gives the possible 
orderings imposed on the last row by a consecutive-ones ordering of $\mT'$, which is
unique up to reversal, by Lemmas~\ref{lem:cycle} and~\ref{lem:Meidanis}.
In each case, if row $i \not\in \{0, 1, k-2, k-1\}$ is chosen to go last,
rows $i-1$ and $i+1$ satisfy the requirements of $A$ and $B$.
$M_T$ has at least five rows and no row of $M_T$ is 
contained in $Z$ more than once in the five iterations, so in at least
one of the iterations a row $i \not\in \{0,1,k-2,k-1\}$ will go last.
\end{proof}

\begin{figure}
\centerline{\includegraphics[]{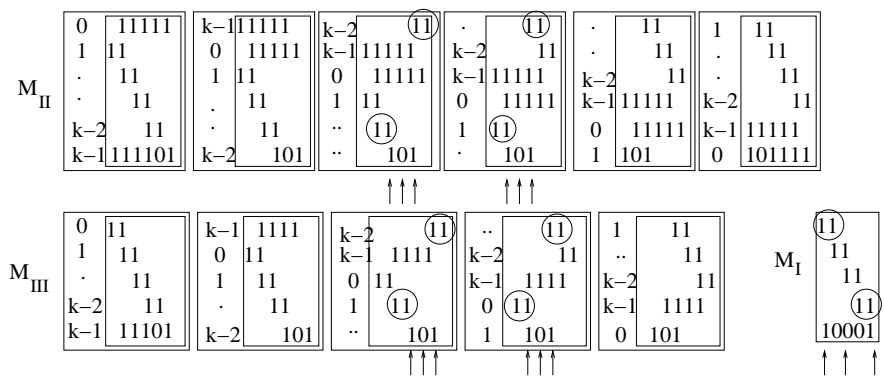}}
\caption{Consecutive-ones orderings of all but the last row of $M_I$, $M_{II}$ and 
$M_{III}$ for different choices of last the last row.  For all but at most four
choices of the last row, there exists a 1-0-1 configuration and rows
$A$ and $B$ satisfying Lemma~\ref{lem:corral}}\label{fig:TuckRotations}
\end{figure}

The correctness and linear time bound for the following are the key results of this
section:

\begin{algo}\label{alg:easyCase}
Find the rows of a Tucker submatrix when every Tucker submatrix has at least five rows
\begin{tabbing}
fin\=dRo\=ws\=($M$) \\
\> Run initialRows($M,4$) (Algorithm~\ref{alg:fiveRows}) to find $\mR'$, $Z$, $A$, $B$ \\
\>\> satisfying the requirements of Lemma~\ref{lem:corral} \\
\> Let $P$ be a shortest path in the overlap graph of $\mR'$ from $A$ to $B$ \\
\> Let $P_1$ be a minimal prefix of $P$ such that the union of $\{Z\}$ and the \\
\>\> set $\mP_1$ of rows of $P_1$ does not have the consecutive-ones property. \\
\> Let $P_2$ be a minimal suffix of $P_1$ such that the union of $\{Z\}$ and the \\
\>\> set $\mP_2$ of rows of $P_2$ does not have the consecutive-ones property. \\
\> Return $\mP_2 \cup \{Z\}$. 
\end{tabbing}
\end{algo}

\begin{lemma}\label{lem:W}
If $M$ does not have the consecutive-ones property and every Tucker matrix of $M$ has at least five rows,
then Algorithm~\ref{alg:easyCase} returns the set of rows of a Tucker matrix of $M$.
\end{lemma}
\begin{proof}
Since $A$ and $B$ lie in the same 
component of the overlap graph of $\mR'$, $P$ exists.
Since $\mR'$ has the consecutive-ones property, so does $\mP$.
Because $\mP$ has a connected overlap graph, $P$, $\bigcup \mP$ is a single
Q node of the PQ tree of $\mP$, by Lemma~\ref{lem:Meidanis}.
Because of $A$ and $B$, $X_h$, $X_i$, and $X_j$ are contained in
distinct Venn classes of $\mP$, and the ones containing $X_h$ and $X_j$ are constrained.
Since $\bigcup \mP$ is consecutive, it must have
a row that contains the $X_i$, hence the Venn class of $\mP$ containing $X_i$ is also constrained.
Therefore, $\mP \cup \{Z\}$ does not have the consecutive-ones property
by Lemma~\ref{lem:violation}, and $P_1$ and $P_2$ exist.
By Lemma~\ref{lem:lastRowTucker},
all Tucker matrices in $\mR' \cup \{Z\}$
contain $Z$, so this applies also to $\mP_2 \cup \{Z\}$.  

Suppose there is a proper subset $\mR''$
of the rows on $P_2$ such that $\mR'' \cup \{Z\}$ contains a Tucker
matrix.  The overlap graph of $\mR''$ is connected, by Lemma~\ref{lem:cycle}.
Since $P_2$ is a shortest path, it is a chordless path,
so $\mR''$ is a subpath of $P_2$ by Lemma~\ref{lem:cycle}.
Let $\mR'_1$ be the rows on $P_1$, excluding the last row on $P_1$.
Let $\mR'_2$ be the rows of $P_2$, excluding the first row on $P_2$.
By the minimality of $P_1$ and $P_2$, $\mR'_1 \cup \{Z\}$ and $\mR'_2 \cup \{Z\}$
have the consecutive-ones
property.  Since $\mR''$ is a subpath of $P_2$,
$\mR'' \subseteq \mR'_1$ or $\mR'' \subseteq \mR'_2$, so $\mR'' \cup \{Z\}$ 
has the consecutive-ones property, contradicting our assumption that it does not.
Therefore, $\mP_2 \cup \{Z\}$ is the set of rows of a Tucker matrix.
\end{proof}

\begin{figure}
\centerline{\includegraphics[]{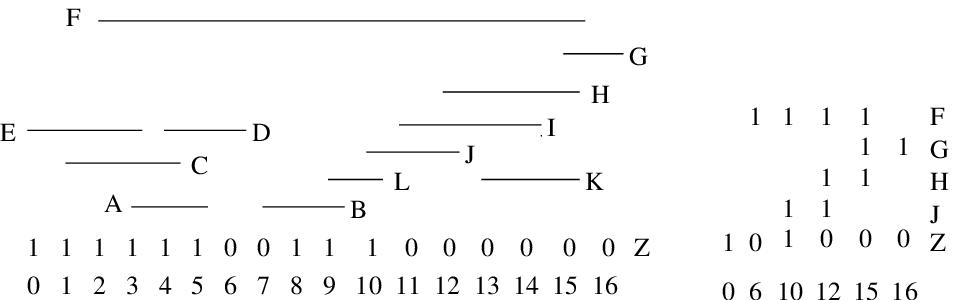}}
\caption{Example of finding a minimal set of rows that does not have
the consecutive-ones property}\label{fig:PP1P2}
\end{figure}

Figure~\ref{fig:PP1P2} gives an example on which we illustrate some
implementation details.    $Z$ is given by the 1's and
0's above the column numbers in the figure on the left, and $\mR'$ is depicted by the intervals.
The rows labeled $A$ and $B$ satisfy the requirements of $A$ and $B$ for Lemma~\ref{lem:corral}, and
$P = (A,E,F,G,H,J,L,B)$ is a shortest path from $A$ to $B$ in the overlap graph
of $\mR'$, found using the BFS algorithm of Section~\ref{sect:BFS}.

Using Booth and Lueker's terminology, we
maintain labels on each class indicating whether it is {\em full}
(contains only 1's of $Z$), {\em empty} (contains only 0's of $Z$), or {\em partial}
(contains both 1's and 0's of $Z$).
The minimal prefix $P_1$ of $P$ whose rows, together with $Z$, do not have
the consecutive-ones property, is $(A,E,F,G,H,J)$.  This is detected as follows.
It is easy to verify that its sequence of constrained Venn classes is
$(\{0,1\}$,$\{2\}$,$\{3\}$, $\{4,5\}$, $\{6-9\}, \{10,11\}$, $\{12\}$, 
$\{13,14\}$, $\{15\}$, $\{16\})$, and their full/partial/empty labels $(F,F,F,F,P,P,E,E,E,E)$, 
respectively.  Selecting a 1 from a full class, a 0 from the first partial class and a 1
from the second partial class gives a 1-0-1 configuration satisfying Lemma~\ref{lem:violation}.
It is the minimal such prefix.
A smaller prefix, $(A,E,F)$ has a 0-1-0 configuration, but it does not
satisfy Lemma~\ref{lem:violation} because $Z \subset A \cup E \cup F$.

The minimal suffix $P_2$ of $P_1$ that satisfies Lemma~\ref{lem:violation} is $(F,G,H,J)$, which
is found in the same way by working on the reverse of $P_1$.  Its sequence
of constrained Venn classes are $(\{2-9\}$, $\{10,11\}$, $\{12-14\}$, $\{15\}$, $\{16\})$,
labeled $(P,P,E, E, E)$, respectively.  Selecting a 0 from the first partial class,
a 1 from the next, and a 0 from an empty class gives a 0-1-0 configuration.  It satisfies
condition~\ref{secondCase} of the lemma, because the unconstrained class, $\{0,1\}$
is partial, hence $Z \not\subseteq F \cup G \cup H \cup J$.

Therefore, $\{F,G,H,J,Z\}$ is the set of rows of a Tucker submatrix.  A minimal set of columns 
that illustrates that it satisfies the lemma is $\{0,6,10,12,15,16\}$.  
On the righthand side of Figure~\ref{fig:PP1P2} is the resulting Tucker matrix, which matches
the final configuration in the sequence for $M_{III}$ in 
Figure~\ref{fig:TuckRotations}.

This example shows that the key to finding $P_1$ and $P_2$ is maintaining
the sequence of constrained Venn classes and their full/partial/empty labels as rows are added in
the order in which they occur on $P$ or on the reverse of $P_1$.
Since they are added in an order such that every prefix of the order has a connected
overlap graph, the sequence is uniquely constrained after each row is added,
by Lemma~\ref{lem:Meidanis}.  When a row $R_i$ is added, if it overlaps a constrained Venn class
$X$, $X$ must be replaced in the sequence with two Venn classes, 
$(X \setminus R_i, X \cap R_i)$ or with $(X \cap R_i, X \setminus R_i)$, whichever
is required to maintain consecutiveness of $R_i$.  If $R_i$ intersects the unconstrained
class, $S$, then $R_i \cap S$ must be added at one extreme end of the sequence, whichever
maintains consecutiveness of $R_i$.  
Details are given in~\cite{McCSODA04}.

The difference between this algorithm and that of~\cite{McCSODA04} (and Booth and Lueker)
is that, instead
of testing at each iteration whether the next row $R_i$ can be added to those
considered so far without undermining the consecutive-ones property, it
must repeatedly perform this test on the fixed row $Z$ after each row $R_i$
is added.  We already know that $R_i$ can be added, since $\mR'$ has the consecutive-ones
property.
Like Booth and Lueker, the previous algorithm of~\cite{McCSODA04} applies the full/partial/empty labels
for $R_i$ to facilitate the test, in $O(|R_i|)$ time, and then removes them before
considering the next row $R_{i+1}$.  
Though we must perform the test on the fixed row $Z$ at each iteration, instead of on $R_i$, we
must do it $O(|R_i|)$ time, not $O(|Z|)$ time, in order to retain the linear time bound.
To do this, we leave the full/partial/empty
labelings for $Z$ from one iteration to the next, so that we only have to update them,
using $R_i$, rather than re-creating them each time a new row is considered.

To facilitate this, we keep updated labels $c(X)$ and 
$n(X)$ on each Venn class $X$, where $c(X)$ denotes the cardinality
of $X$ and $n(X)$ is the number of elements of $Z$ in $X$.  Labels only need to be updated
when a Venn class is split.  
It is split into $X \cap R_i$ and $X \setminus R_i$.   We may
find $c(X \cap R_i)$ and $n(X \cap R_i)$ by counting them directly, since there are
$O(|R_i|)$ of these elements.  The classes are implemented with doubly-linked lists,
and these sets are removed from the list for $X$, leaving it to represent $X \setminus R_i$.
Subtracting $c(X \cap R_i)$ and $n(X \cap R_i)$ from the old labels $c(X)$ and $n(X)$
gives the updated labels for $c(X \setminus R_i)$ and $n(X \setminus R_i)$ in $O(1)$ time.
Each of the new classes is full if its $c()$ and $n()$ labels are equal,
empty if its $n()$ label is 0, and partial otherwise.

To evaluate whether one of the conditions of Lemma~\ref{lem:violation} holds, it is easy to see
that it suffices to keep track of {\em transition pairs},
which are consecutive pairs such that one contains a 0 and one contains a 1.
This happens when their full/partial/empty labels are unequal, or else both
partial.  When a new transition pair forms, we have touched at least one member
of the pair within our $O(|R_i|)$ operations, so keeping track of these
does not affect this time bound.

Since finding $P$ takes linear time by the BFS of Section~\ref{sect:BFS}, it
remains only to bound the time required to find the first step,
finding the elements $A$ and $B$ of Lemma~\ref{lem:corral}.   This is
much more straightforward, since we apply it once for each iteration of the loop
of Algorithm~\ref{alg:fiveRows}, hence we can afford to take $\Theta(size(M))$
time for the test.  We can apply the entire set of full/partial/empty 
labels for $Z$ to the PQ tree within this bound, by working from the leaves
to the root.

For each Q node $Q$, the members of the overlap component whose union is $Q$
are unions of more than one and fewer than all of its children.  How to find them
in linear time for all Q nodes has been described previously, for example in~\cite{LB79}.
We find the rows of the overlap component that contain a Venn child of $Q$
that is labeled full or partial (a ``1'').  Out of all such rows, let $A'$ be the one with a
leftmost right endpoint, and let $B'$ be the one with the rightmost
left endpoint.  By a simple greedy swapping argument, 
the overlap component giving rise to $Q$ contains an $A$ and a $B$ satisfying
Lemma~\ref{lem:corral} if and only if $A'$ and $B'$ satisfy it, which happens
if and only if there is a child between
the right endpoint of $A'$ and the left endpoint of $B'$ that is labeled
partial or empty (a ``0'').
This can also clearly be implemented so that the time bound over all 
Q nodes takes $O(size(M))$ time. 

Once the set $\mP_2 \cup \{Z\}$ of rows of a Tucker matrix have been found, it remains to find
the columns.  This is a set of columns, the removal of any one of which
would undermine the conditions of Lemma~\ref{lem:violation}, which are satisfied
initially by $\mP_2 \cup \{Z\}$.  Deletion of a column undermines the lemma if and only if
does one of the following:

\begin{enumerate}
\item It disconnects the path in overlap graph on $\mR$;
\item It undermines the only remaining 1-0-1 configuration, or 0-1-0 configuration with a 1 in the
unconstrained class.
\end{enumerate}

The second test is elementary and omitted because of space constraints.
For the first test, recall that $P_2 = (R_1, R_2, \ldots, R_k)$ is a chordless
path in the overlap graph.
Let $\mA = \{R_1 \setminus R_2\} \cup \{X| X = R_i \cap R_{i+1}$ for $i \in \{1,2, \ldots, k-1\}\}$
$ \cup \{Y| Y = R_{i+1} \setminus R_i$ for $i \in \{1,2, \ldots, k-1\}\}$.
Each element of $\mA$ is consecutive-ones ordered,
the sum of cardinalities of sets in $\mA$ is $O(size(M))$.
The overlap graph remains connected if and only if every member of $\mA$ contains
at least one retained column.  We give each column a list of members of $\mA$ it is contained in and keep
a counter on each element of $\mA$ indicating the number of remaining columns it contains.
When removing a column $C$, the counters can be updated by decrementing the counters
of members of $\mA$ in its list.  A column cannot be removed if removing it would decrement
a counter to 0.

\section{Finding a Lekkerkerker-Boland Subgraph}

Tucker observed that the smallest graphs whose clique matrices
contain a Tucker matrix must be exactly the LB
graphs~\cite{Tuck72}.  However, making use of this to find an
LB subgraph is not as straightforward as it appears.
It is easy to believe from Figure~\ref{fig:LBGraphs} that 
that the rows of every Tucker submatrix in a clique matrix
can be extended to a clique matrix 
of an LB subgraph by including at most three additional rows.
The rows of the result would therefore identify the vertices that induce
an LB subgraph.

That this reasoning is flawed
is illustrated by Figure~\ref{fig:LBExtension}.  Neither does it follow
from the fact that no column of the clique matrix is a subset of any other;
we discovered the example in the figure when trying to prove it using only this fact.
Therefore, though the conclusion is true in the case of chordal graphs, this
requires proof.

\begin{figure}
\centerline{\includegraphics[]{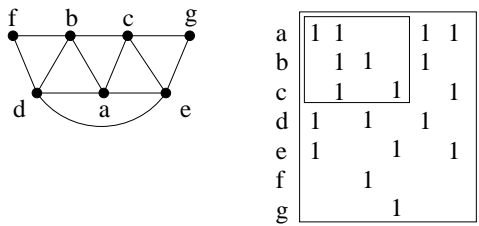}}
\caption{A Tucker matrix, $M_{III}$,  in the clique matrix of a graph $G$.
The only LB subgraph is the chordless
cycle $(b,c,e,d)$, which does not contain row $a$ of the $M_{III}$.}\label{fig:LBExtension}
\end{figure}

If $G$ is not chordal, we may return a $G_{III}$ by the algorithm of~\cite{TarjYan85}.
Henceforth, we may assume that the graph is chordal.  
A {\em clique tree} of a chordal graph is a tree that has one node
for each maximal clique, and with the property that for each
vertex $v$ of $G$, the cliques that contain $v$ induce a connected
subtree.
Every chordal graph has a clique tree, see for example~\cite{Gol80}.
A vertex is {\em simplicial} if its neighbors induce a complete subgraph.
The following are immediate from results that appear there:

\begin{lemma}\label{lem:leafDeletion}
Let $T$ be a clique tree for a chordal graph $G$, and let $K$ be a leaf.  Then $K$ contains
a simplicial vertex of $G$.  Let $S$ be the simplicial vertices of $K$, and
let $T'$ be the result of deleting leaf $K$ from $T$. 
Deleting $S$ from $G$ yields an induced subgraph that has
$T'$ as a clique tree.
\end{lemma}

\begin{definition}\label{def:shrinking}
By {\em shrinking a clique tree $T$}, let us denote
the operation of deleting the set $S$ of simplicial vertices in a leaf
$K$ of $T$, yielding a graph with the smaller clique tree described by the lemma.
\end{definition}

\begin{lemma}\label{lem:cliqueConnected}
If $G$ is a chordal graph, $T$ is a clique tree of $G$, and $G[X]$ is connected,
then the cliques of $T$ that contain members of $X$
induce a connected subtree of $T$.
\end{lemma}

\begin{lemma}\label{lem:simplicialClique}
Let $T$ be a clique tree of a chordal graph $G$,
$\mK$ a collection of cliques of $G$, and a clique $C \in \mK$.
If every member of $\mK - C$ contains a vertex that $C$ does not have
and $G[\bigcup \mK \setminus C]$ is connected,
then $C$ does not lie on the path in $T$ between any pair of cliques of $\mK - C$.
\end{lemma}

\begin{definition}\label{def:specialCols}
If a Tucker matrix occurs as a submatrix induced by column set $\mC$
and row set $\mR$ of a clique matrix, let 
a {\em private row} for $C \in \mC$ be a row $R \not\in \mR$
that is contained in no column of $\mC$ other than $C$.
Let the {\em special columns} of 
$M_{II} - M_{V}$ be those that are subsets of other columns of these matrices,
and let the {\em special columns} of an instance of $M_I$ on three rows be
all three columns.
\end{definition}

The special columns are the ones
that Tucker identified as ``asteroidal column triples'' in the bipartite
incidence graph of rows and and columns in~\cite{Tuck72}.  However, the existence of simplicial vertices
or ``asteroidal vertex triples'' in a a sense defined in~\cite{LB62} was not examined.
The paper dealt with arbitrary consecutive-ones matrices where
such rows need not occur.

\begin{lemma}\label{lem:TuckerSimplicials}
Let $G$ be a chordal graph and let $M_T$ be a submatrix of a clique matrix
that is an instance of $M_I$ on three vertices or
an instance of $M_{II} - M_V$.  Then the clique matrix has a private row
for each special column of $M_T$.
\end{lemma}
\begin{proof}
Let $\mK$ be the set of cliques of $G$ that contain the columns of this instance of $M_T$,
let $C \in \mK$ contain a special column of the instance, and let $T$ be a clique tree
for $G$.
In each case, $C$ does not lie on the unique path in $T$
between any pair of members of $\mK - C$, by Lemma~\ref{lem:simplicialClique}.
Therefore, iteratively shrinking $T$ (Definition~\ref{def:shrinking})
until it cannot be further shrunk without deleting a clique of $\mK$ results in an
induced subgraph $G'$ with a clique tree $T'$ where $C$ is a leaf.  By Lemma~\ref{lem:leafDeletion},
$C$ has a simplicial vertex, which must be a private row for $C$ in $\mK$.
\end{proof}

By examination, it is easy to verify in the case of $M_I$ on three rows,
and $M_{II} - M_{IV}$ that adding the private rows for the special columns
to the rows and columns of the submatrix gives an interval model for
the corresponding LB subgraph.  The resulting set of rows therefore induces
an LB subgraph.

\bibliographystyle{plain}
\bibliography{to}

\begin{thebibliography}{10}

\bibitem{BLS99}
A.~Brandstaedt, V.B. Le, and J.P. Spinrad.
\newblock {\em Graph Classes: A Survey}.
\newblock SIAM Monographs on Discrete Mathematics, Philadelphia, 1999.

\bibitem{CCHS10}
C.~Chauve, U.-U. Haus, T.~Stephen, and V.~P. You.
\newblock Minimal conflicting sets for the consecutive ones property in
  ancestral genome reconstruction.
\newblock {\em Journal of Computational Biology}, 17:1167--1181, 2010.

\bibitem{Gol80}
M.~C. Golumbic.
\newblock {\em Algorithmic Graph Theory and Perfect Graphs}.
\newblock Academic Press, New York, 1980.

\bibitem{HellUIG}
P.~Hell and J.~Huang.
\newblock Certifying {L}ex{BFS} recognition algorithms for proper inteval
  graphs and proper interval bigraphs.
\newblock {\em SIAM J. Discrete Math}, 18:554--570, 2004.

\bibitem{KMMSJournal}
D.~Kratsch, R.M. McConnell, K.~Mehlhorn, and J.P. Spinrad.
\newblock Certifying algorithms.
\newblock {\em SIAM Journal on Computing}, 36:236--353, 2006.

\bibitem{LB62}
C.~Lekkerker and D.~Boland.
\newblock Representation of finite graphs by a set of intervals on the real
  line.
\newblock {\em Fund. Math.}, 51:45--64, 1962.

\bibitem{LB79}
G.~S. Lueker and K.~S. Booth.
\newblock A linear time algorithm for deciding interval graph isomorphism.
\newblock {\em J. ACM}, 26:183--195, 1979.

\bibitem{McCSODA04}
R.~M. McConnell.
\newblock A certifying algorithm for the consecutive-ones property.
\newblock {\em Proceedings of the 15th Annual ACM-SIAM Symposium on Discrete
  Algorithms (SODA04)}, 15:761--770, 2004.

\bibitem{MMNSCert}
R.~M. McConnell, K.~Mehlhorn, S.~N{\"{a}}her, and P.~Schweitzer.
\newblock Certifying algorithms.
\newblock {\em Computer Science Reviews}, 5:119--161, 2011.

\bibitem{MeidanisPT98}
J.~Meidanis, O.~Porto, and G.P. Telles.
\newblock On the consecutive ones property.
\newblock {\em Discrete Applied Mathematics}, 88:325--354, 1998.

\bibitem{Roberts78}
Fred~S. Roberts.
\newblock {\em Graph Theory and Its Applications to Problems of Society}.
\newblock Society for Industrial and Applied Mathematics, Philadelphia, 1978.

\bibitem{RTL:triangulated}
D.~Rose, R.~E. Tarjan, and G.~S. Lueker.
\newblock Algorithmic aspects of vertex elimination on graphs.
\newblock {\em SIAM J. Comput.}, 5:266--283, 1976.

\bibitem{SpinGR}
J.~Spinrad.
\newblock {\em Efficient Graph Representations}.
\newblock American Mathematical Society, Providence RI, 2003.

\bibitem{StoyeWittler09}
J~Stoye and R.~Wittler.
\newblock A unified approach for reconstructing ancience gene clusters.
\newblock {\em IEEE/ACM Transactions on Computational Biology and
  Bioinformatics}, 6:387--400, 2009.

\bibitem{TarjYan85}
E.~E. Tarjan and M.~Yannakakis.
\newblock Addendum: Simple linear-time algorithms to test chordality of graphs,
  test acyclicity of hypergraphs, and selectively reduce acyclic hypergraphs.
\newblock {\em SIAM Journal on Computing}, 14:254--255, 1985.

\bibitem{Tuck72}
A.~Tucker.
\newblock A structure theorem for the consecutive 1's property.
\newblock {\em Journal of Combinatorial Theory, Series B}, 12:153--162, 1972.

\bibitem{Wegner67}
G.~Wegner.
\newblock {\em Eigenschaften der {N}erven {H}omologishe-{E}infactor {F}amilien
  in $R^n$}.
\newblock PhD thesis, Universit{\"{a}}t G{\"{o}}ttingen, 1967.

\end{thebibliography}
\end{document}